\documentclass[a4paper,USenglish,autoref,cleveref,thm-restate]{lipics-v2021}

\pdfoutput=1
\usepackage{xspace}
\usepackage{mathtools}
\usepackage{delimset}
\usepackage{algorithm}
\usepackage{algpseudocodex}

\EventEditors{Olaf Beyersdorff, Micha\l{} Pilipczuk, Elaine Pimentel, and Nguyen Kim Thang}
\EventNoEds{4}
\EventLongTitle{42nd International Symposium on Theoretical Aspects of Computer Science (STACS 2025)}
\EventShortTitle{STACS 2025}
\EventAcronym{STACS}
\EventYear{2025}
\EventDate{March 4--7, 2025}
\EventLocation{Jena, Germany}
\EventLogo{}
\SeriesVolume{327}
\ArticleNo{27}

\nolinenumbers

\title{Online Disjoint Set Covers: Randomization is not Necessary}

\author{Marcin Bienkowski}{University of Wroc{\l}aw, Poland}{marcin.bienkowski@cs.uni.wroc.pl}{https://orcid.org/0000-0002-2453-7772}{}
\author{Jarosław Byrka}{University of Wroc{\l}aw, Poland}{jaroslaw.byrka@cs.uni.wroc.pl}{https://orcid.org/0000-0002-3387-0913}{}
\author{Łukasz Jeż}{University of Wroc{\l}aw, Poland}{lukasz.jez@cs.uni.wroc.pl}{https://orcid.org/0000-0002-7375-0641}{}

\authorrunning{M. Bienkowski, J. Byrka and Ł. Jeż}
\Copyright{Marcin Bienkowski, Jarosław Byrka and Łukasz Jeż}
\ccsdesc{Theory of computation~Online algorithms}
\keywords{Disjoint Set Covers, Derandomization, pessimistic Estimator, potential Function, online Algorithms, competitive Analysis}

\funding{Supported by Polish National Science Centre grants 2022/45/B/ST6/00559, \\ 
2020/39/B/ST6/01641, and 2020/39/B/ST6/01679.}

\DeclarePairedDelimiter{\ceil}{\lceil}{\rceil}

\DeclarePairedDelimiter{\set}{\lbrace}{\rbrace}

\renewcommand{\leq}{\leqslant}
\renewcommand{\geq}{\geqslant}
\newcommand{\OPT}{\textsc{Opt}\xspace}
\newcommand{\ALG}{\textsc{Alg}\xspace}
\newcommand{\DET}{\textsc{Det}\xspace}
\newcommand{\RAND}{\textsc{Rand}\xspace}
\newcommand{\I}{E}
\newcommand{\e}{\mathrm{e}}
\newcommand{\E}{\mathbb{E}}
\newcommand{\eps}{\varepsilon}
\newcommand{\plow}{\ensuremath{p_S}}

\newcommand{\T}{\mathcal{T}_E}
\newcommand{\R}{\mathcal{R}}

\definecolor{blue}{RGB}{0,50,200}
\definecolor{magenta}{RGB}{255,0,255}
\hypersetup{colorlinks,linkcolor=blue,urlcolor=blue,citecolor=magenta}

\begin{document}
\maketitle

\DeclareFontShape{T1}{lmr}{m}{scit}{<->ssub * lmr/m/scsl}{}


\begin{abstract}
In the online disjoint set covers problem, the edges of a hypergraph are
revealed online, and the goal is to partition them into a maximum number of
disjoint set covers. That is, $n$ nodes of a hypergraph are given at the
beginning, and then a sequence of hyperedges (subsets of $[n]$) is presented to
an~algorithm. For each hyperedge, an online algorithm must assign a color (an
integer). Once an~input terminates, the gain of the algorithm is the number of
colors that correspond to valid set covers (i.e., the union of hyperedges that
have that color contains all $n$~nodes).

We present a deterministic online algorithm that is $O(\log^2 n)$-competitive,
exponentially improving on the previous bound of $O(n)$ and matching the
performance of the best randomized algorithm by Emek et al.~[ESA 2019].

For color selection, our algorithm uses a novel potential function, which can be
seen as an online counterpart of the derandomization method of conditional
probabilities and pessimistic estimators. There are only a few cases where
derandomization has been successfully used in the field of online algorithms. In
contrast to previous approaches, our result extends to the 
following new challenges: (i) the potential function derandomizes not only the
Chernoff bound, but also the coupon collector's problem, (ii) the value of \OPT
of the maximization problem is not bounded a~priori, and (iii) we do not produce
a fractional solution first, but work directly on the input. 
\end{abstract}


\section{Introduction}

In this paper, we study online algorithms for maximizing the number of set
covers of a set of nodes. We focus on a hypergraph (set system) $G = (V,E)$ that
has $n = |V|$ nodes and where each hyperedge $S \in E$ is a subset of nodes from
$V$. A subset $E' \subseteq E$ is called \emph{set cover} if $\bigcup_{S \in E'}
S = V$, i.e., every node is covered by at least one hyperedge of $E'$. In the
\emph{disjoint set covers} (DSC) problem \cite{FeHaKS02,CarDu05,PaBaVa15}, the
goal is to partition the set of hyperedges $E$ into \emph{maximum number} of
mutually disjoint subsets $E = E_1 \uplus E_2 \uplus \dots \uplus E_k$, where
each $E_j$ is a set cover. Note that $E$ is a multi-set, i.e., it can contain
multiple copies of the same hyperedge.

The problem has been studied in a theoretical setting, but as we discuss later,
it also finds applications in sensor networks~\cite{CarDu05} or assignment
tasks~\cite{PaBaVa15}.

\subparagraph{Coloring Perspective.}

When constructing a solution to the DSC problem, it is convenient to think in
terms of coloring hyperedges.\footnote{The DSC problem should not be confused
with the hyperedge coloring problem, which requires that the hyperedges of the
same color be disjoint.} Each color then corresponds to a~subset of hyperedges
that have that color, and the color is \emph{fully used} if the hyperedges
colored with it form a set cover. The problem then becomes equivalent to
coloring hyperedges, so that the number of fully used colors is maximized.

\subparagraph{Online Variant.}

In this paper, we focus on the \emph{online} variant of the DSC problem, where
the set $V$ is given in advance, but the hyperedges of $E$ arrive in an online
fashion. Once a~hyperedge~$S \in E$ arrives, it must be colored immediately and
irrevocably. Again, the goal is to maximize the number of fully used colors. The
performance of an online algorithm is measured by the competitive ratio, i.e.,
the ratio of the number of fully used colors produced by the optimal
\emph{offline} algorithm $\OPT$ to that of an online algorithm $\ALG$.

\subparagraph{Our Contribution.}

We present a deterministic online $O(\log^2 n)$-competitive algorithm \DET for
the DSC problem, which exponentially improves on the $O(n)$-competitive
algorithm by Emek et al.~\cite{EmGoKa19} and matches the performance of their
randomized algorithm~\cite{EmGoKa19}. We discuss the challenges and technical
contribution in more detail in \autoref{sec:our_results}.


\subsection{Offline Scenario: Previous Results}
\label{sec:offline_scenario}

The disjoint set covers problem is a fundamental NP-complete
problem~\cite{CarDu05} that can be approximated within a~factor of $(1+o(1))
\cdot \ln |V|$~\cite{PaBaVa15} and cannot be approximated within a~factor of
$(1-\eps) \cdot \ln |V|$ for $\eps > 0$ unless $\textsf{NP} \subseteq
\textsf{DTIME}(n^{\log \log n})$~\cite{FeHaKS02}.\footnote{The authors of
\cite{FeHaKS02} call this problem \emph{set cover packing} or \emph{one-sided
domatic problem}.}

\subparagraph{OPT vs. Min-degree.}

We use $\OPT(E)$ to denote the maximum number of disjoint set covers of a
hypergraph $G=(V,E)$. This value is also called \emph{cover-decomposition
number}~\cite{BoPrRS13}. We denote the minimum degree of the hypergraph
$G=(V,E)$ as $\delta(E) \triangleq \min_{i \in V} |S \in E \,:\, i \in S|$. Note
that trivially $\OPT(E) \leq \delta(E)$. While this bound may not be tight
(cf.~\autoref{sec:related}), $\delta(E)$ serves as a natural benchmark for
approximation and online algorithms.

\subparagraph{Random Coloring and its Straightforward Derandomization.}

An $O(\log n)$-approximation algorithm for the \emph{offline} DSC
problem~\cite{PaBaVa17} can be obtained by coloring each hyperedge with a color
chosen uniformly at random from the palette of $\Theta(\delta(E) / \log n)$
colors. To analyze this approach, we focus on a single node $i \in V$. We say that
node $i$ \emph{gathers} color $r$ if $i$ is contained in a~hyperedge colored with
$r$. Since node $i$ is contained in at least $\delta(E)$ hyperedges, and there
are $\Theta(\delta(E) / \log n)$ available colors, $i$ gathers all palette
colors with high probability. By the union bound, this property holds for all
nodes, i.e., all $\Theta(\delta(E) / \log n)$ colors are fully used by an
algorithm (also with high probability). Since $\OPT(E) \leq \delta(E)$, the
approximation ratio of $O(\log n)$ follows.

The hyperedges can be processed in a fixed order, and the random choice of
a~color can be replaced by the deterministic one by a simple application of the
method of conditional probabilities~\cite{PaBaVa17,AloSpe00}.


\subsection{Online Scenario: Previous Results}
\label{sec:online_scenario}

In the online case, an algorithm first learns the set of nodes $V$, and then the
edges of $E$ are revealed sequentially. For notational convenience, we use $E$
to denote both the set and a~sequence of hyperedges (the input). We use
$\ALG(E)$ to denote the number of fully used colors in the solution of an online
algorithm $\ALG$.

It is important to note that no parameter of the hypergraph other than the
number of nodes is known a priori. In particular, an online algorithm does not
know the value of $\delta(E)$ in advance.

\subparagraph{Competitive Ratio.}

We say that $\ALG$ is \emph{(non-strictly) $c$-competitive} if there exists 
$\beta \geq 0$ such that 
\begin{equation}
\label{eq:competitive_ratio}
    \textstyle \ALG(E) \geq \OPT(E) \,/\, c - \beta.
\end{equation}
While $\beta$ can be a function of $n$, it cannot depend on the input sequence
$E$. If \eqref{eq:competitive_ratio} holds with $\beta = 0$, 
the algorithm is \emph{strictly $c$-competitive}.

For randomized algorithms, we replace $\ALG(E)$ with its expected value taken
over the random choices of the algorithm.

\subparagraph{Randomized Algorithms.}

The current best randomized algorithm was given by Emek et
al.~\cite{EmGoKa19}; it achieves the (strict) competitive ratio of $O(\log^2 n)$. The
general idea of their algorithm is as follows. To color a hyperedge $S$, their
algorithm first computes the minimum degree~$\delta^*$ of a node from $S$. By a
combinatorial argument, they show that they can temporarily assume that
$\delta(E) = O(n \cdot \delta^*)$. Their algorithm then chooses $\ell$ uniformly
at random from the set $\set{ \delta^*, 2 \delta^*, 4 \delta^*, \dots, 2^{\log
n} \cdot \delta^* }$; with probability $\Omega(1/\log n)$ such $\ell$ is 
a~2-approximation of~$\delta(E)$. Finally, to color $S$, it chooses a color 
uniformly at random from the palette of $\Theta(\ell / \log^2 n)$ colors, 
using the arguments similar to those in the offline case.

We refrain from discussing it further here, as we present a variant of their
algorithm, called \RAND, in \autoref{sec:algorithm_definition} (along with a
description of the differences from their algorithm). 

The best lower bound for the (strict and non-strict) competitive ratio of a
randomized algorithm is $\Omega(\log n / \log \log n)$~\cite{EmGoKa19}; it
improves on an~earlier bound of $\Omega(\sqrt{\log n})$ by Pananjady et al.
\cite{PaBaVa15}.

\subparagraph{Deterministic Algorithms.}

The deterministic case is well understood if we restrict our attention to
\emph{strict} competitive ratios. In such a case the lower bound  
is $\Omega(n)$~\cite{PaBaVa15}. The asymptotically matching upper bound $O(n)$
is achieved by a~simple greedy algorithm~\cite{EmGoKa19}.

On the other hand, the current best lower bound for the non-strict
deterministic competitive ratio is $\Omega(\log n / \log \log
n)$~\cite{EmGoKa19}.\footnote{This discrepancy in achievable strict and
non-strict competitive ratios is quite common for many maximization problems:
the non-strict competitive ratio allows the algorithm to have zero gain on very
short sequences, thus avoiding initial choices that would be bad in the long
run.} The $O(n)$ upper bound of~\cite{EmGoKa19} clearly holds also in the
non-strict setting, but no better algorithm has been known so far.

\begin{table}[t]
\centering
\begin{tabular}{c|c|c}
\hline
& \textbf{Upper bound} & \textbf{Lower bound} \\
\hline
randomized & \multirow{2}{*}{$O(\log^2 n)$~\cite{EmGoKa19}} & \multirow{2}{*}{$\Omega(\log n / \log \log n)$~\cite{EmGoKa19}} \\
strict and non-strict & & \\
\hline
deterministic & \multirow{2}{*}{$O(n)$~\cite{EmGoKa19}} & \multirow{2}{*}{$\Omega(n)$~\cite{PaBaVa15}}  \\ 
strict & & \\
\hline
deterministic
 & $O(n)$~\cite{EmGoKa19} & \multirow{2}{*}{$\Omega(\log n / \log \log n)$~\cite{EmGoKa19}}  \\ 
non-strict
 & $\mathbf{O(log^2\, n)}$~(\autoref{thm:main}) & \\
\hline
\end{tabular}
\caption{Previous and new bounds on the strict and non-strict competitive ratios for the 
online DSC problem, for randomized and deterministic algorithms.}
\label{tab:results}
\end{table}

\subparagraph{Lack of General Derandomization Tools.}

Unlike approximation algorithms, derandomization is extremely rare in online
algorithms.\footnote{Many online problems (e.g., caching~\cite{Young02,AdCzER19}
or metrical task systems~\cite{BoLiSa92,BuCoRa23,BuCoLL19}) exhibit a provable
exponential discrepancy between the performance of the best randomized and
deterministic algorithms. For many other problems, the best known deterministic
algorithms are quite different (and more complex) than randomized ones.} To
understand the difficulty, consider the standard (offline) method of conditional
probabilities~\cite{AloSpe00} applied to the offline DSC problem: the resulting
algorithm considers all the random choices (color assignments) it could make for
a~given hyperedge $S$, and computes the probability that future random choices,
\emph{conditioned on the current one}, will lead to the desired solution.
Choosing the color that maximizes this probability ensures that the probability
of reaching the desired solution does not decrease as subsequent hyperedges are
processed. In some cases, exact computations are not possible, but the algorithm
can instead estimate this probability by computing a so-called \emph{pessimistic
estimator}~\cite{Raghav88}. This is not feasible in the online setting, since an
algorithm does not know the future hyperedges, and thus cannot even estimate
these probabilities.\footnote{As observed by Pananjady et al.~\cite{PaBaVa15},
however, these probabilities can be estimated if an online algorithm knows the
final min-degree~$\delta(E)$ in advance, which would lead to an $O(\log
n)$-competitive algorithm for this semi-offline scenario.}


\subsection{Our Technical Contribution}
\label{sec:our_results}

In this paper, we present a deterministic online polynomial-time algorithm \DET
that is $O(\log^2 n)$-competitive, which exponentially improves the previous
bound of $O(n)$. This resolves an open question posed by the authors
of~\cite{EmGoKa19} who asked whether the method of conditional expectations
could be used to derandomize their algorithm. 

Our bound is obtained for the non-strict competitive ratio: as we discussed in
the previous subsection, this is unavoidable for the DSC problem. Our result
requires a relatively small~($1/4$) additive term~$\beta$ in the definition of the
competitive ratio \eqref{eq:competitive_ratio}.

We begin by constructing a \emph{randomized} solution \RAND. It will be
a variation of the approach by Emek et al.~\cite{EmGoKa19}; we present it in
\autoref{sec:algorithm_definition}. As \RAND is closely related to the previous
randomized algorithm, it is quite plausible that it achieves the same
competitive ratio of $O(\log^2 n)$. However, our analysis does not support this
conjecture. Instead, we use \RAND as a stepping stone to our deterministic
algorithm in the following way.

\begin{itemize}
\item We define a particular random event, denoted $\T$, of \RAND executed on
    instance $\I$. 
\item We show that for each execution of \RAND satisfying $\T$, it holds that
    $\RAND(\I) \geq \Omega(1/\log^2 n) \cdot \OPT(\I) - 1/4$. 
\end{itemize}

We will provide the exact definition of the event $\T$ in
\autoref{sec:introducting_potential}. For now, we just note that $\T$ is a
property that certifies that, \emph{at each step $t$}, we can relate the number
of colors gathered so far by each node $i$ to its current degree. 

While a relaxed version of the property $\T$ (requiring that the relation
between the gathered colors and the current degree \emph{only at the end of the
input}) follows quite easily with a constant probability, it seems that $\T$
itself is quite strong and does not hold with a~reasonable probability.
However, this is not an obstacle to creating a deterministic solution, as we
show in the following claim.
\begin{itemize}
    \item It is possible to replace the random choices of \RAND by deterministic
    ones (in an online manner), so that $\T$ always holds. 
\end{itemize}
This will immediately imply the competitive ratio of the resulting deterministic
algorithm.

Our approach is based on a novel potential function that guides the choice of
colors. As we discuss later in \autoref{sec:introducting_potential}, our
approach can be seen as an~online-computable counterpart of the method of
conditional probabilities that simultaneously controls the derandomization of
the Chernoff bound and the coupon collector's problem.


\subsection{Related Work}
\label{sec:related}

Applications of the DSC problem include allocating servers to users in file
systems and users to tasks in crowd-sourcing platforms~\cite{PaBaVa15}. 

Another application arises in a sensor network, where each node corresponds to
a~monitoring target and each hyperedge corresponds to a sensor that can monitor
a set of targets. At any given time, all targets must be monitored. A possible
battery-saving strategy is to partition the sensors into disjoint groups (each
group covering all targets) and activate only one group at a time, while the
other sensors remain in a low-power mode~\cite{CarDu05}. 

When the assumption that each sensor can only participate in a single group is
dropped, this leads to more general sensor coverage problems. The goal is then
to maximize the lifetime of the network of sensors while maintaining the
coverage of all targets. The offline variants of this problem have been studied
both in general graphs~\cite{BeCSZ04,CaThLW05,PaBaVa17} and in geometric
settings~\cite{BuEJVY07,GibVar11}.

Another line of work studied the relationship between the minimum degree
$\delta(E)$ and the cover-decomposition number $\OPT(E)$. As pointed out in
\autoref{sec:offline_scenario}, $\delta(E)/\OPT(E) \geq 1$. This ratio is
constant if $G$ is a graph~\cite{Gupta78}, and it is at most $O(\log n)$ for
every hypergraph; the latter bound is asymptotically
tight~\cite{BoPrRS13}. Interestingly enough, all the known papers on the DSC
problem (including ours) relate the number of disjoint set covers to
$\delta(E)$. It is an~open question whether our estimate of the competitive
ratio is tight: it could possibly be improved by relating the gain of an
algorithm directly to $\OPT(E)$ instead of $\delta(E)$.


\subsection{Preliminaries}

An input to our problem is a gradually revealed hypergraph $G=(V,E)$. $V$
consists of $n$ nodes numbered from $1$ to $n$, i.e., $V = [n]$. The set $E$ of
hyperedges is presented one by one: in a step $t$, an algorithm learns $S_t \in
E$, where $S_t \subseteq [n]$, and has to color it. We say that node $i$
\emph{gathers} color $r$ if $i$ is contained in a hyperedge colored with $r$. We
say that a~color $r$ is \emph{fully used} if all nodes have gathered it. The
objective of an~algorithm is to maximize the number of fully used colors. 

At any point in time, the \emph{degree} of a node $j$, denoted $\deg(j)$, is the
current number of hyperedges containing~$j$. Let $\delta(E) = \min_{i \in [n]}
|\set{ S_t \in E : v_i \in S_t }|$, i.e., $\delta(E)$ is the minimum degree of
$G$ at the end of the input $E$. Clearly, $\OPT(E) \leq \delta(E)$. It is
important to note that $\delta(E)$ is not known in advance to an online
algorithm.


\section{Our Algorithm}

\subsection{Definition of RAND}
\label{sec:algorithm_definition}

We start with some notions, defined for each integer $k \geq 0$:
\begin{align*}
    h \triangleq \ceil{ \log n }, 
    && q_k \triangleq \left\lceil \brk*{ 1 - \frac{1}{2n} } \cdot 2^k \right\rceil, 
    && R_k \textstyle \triangleq \left\{ 2^k, \dots, 2^{k+1} - 1 \right\}. 
\end{align*}
Note that $|R_k| = 2^k$. 

\begin{algorithm}[t]
    \caption{Definition of \RAND for a hyperedge $S$ in step $t$}
    \label{alg:rand}
    \begin{algorithmic}[1]
    
    \State{$\triangleright$ \emph{Initialization}}
    \For{each node $i \in [n]$}
        \State $p(i) \gets 0$
            \Comment{all nodes start in phase $0$}
        \For{each integer $k \geq 0$} \Comment{and they have no colors yet}
            \State $C_{i,k} \gets \emptyset$                 
        \EndFor
    \EndFor
    
    \Statex
    \State{$\triangleright$ \emph{Runtime}}
    \For{each hyperedge $S$ appearing in sequence $\I$}
        \State $\plow \gets \min_{i \in S} p(i)$
        \State Sample $k^*$ from $\set {\plow, \plow+1, \dots, \plow + h - 1 }$
            \Comment{uniform distribution}
        \State Sample color $r$ from $R_{k^*}$ 
            \Comment{uniform distribution}
        \For{each node $i \in S$}
            \State $C_{i,k^*} \gets C_{i,k^*} \cup \set{r}$
                \label{line:C_update}
            \If{$|C_{i,p(i)}| \geq q_{p(i)}$}   
                \Comment{end node phase if it gathered $q_{p(i)}$ colors}
                \State{$p(i) \gets p(i)+1$}
            \EndIf                
        \EndFor
    \EndFor
    \end{algorithmic}
\end{algorithm}

For each node $i$ independently, $\RAND$ maintains its \emph{phase} $p(i)$,
initially set to $0$. We use $C_{ik}$ to denote the set of colors from the
palette $R_k$ that node $i$ has gathered so far. A~phase~$k$ for node $i$
ends when it has gathered $q_k$ colors from palette $R_k$. In such a case,
node $i$ increments its phase number $p(i)$ at the end of the step.

We now describe the behavior of \RAND in a single step, when a hyperedge $S$
appears. Let $\plow = \min_{i \in S} p(i)$, where $p(i)$ is the phase of
node $i$ before $S$ appears. \RAND first chooses a random integer $k^*$
uniformly from the set $\set{ \plow, \plow+1, \dots, \plow + h - 1}$. Second,
\RAND chooses a color $r$ uniformly at random from the set $R_{k^*}$ and colors
$S$ with it; in effect all nodes in $S$ gather color $r$.

The pseudocode of \RAND is given in \autoref{alg:rand}.

\subparagraph{Comparison with the Previous Randomized Algorithm.}

\RAND is closely related to the randomized algorithm by Emek et
al.~\cite{EmGoKa19}. The main difference is that the phases of the nodes in
their algorithm are of fixed lengths, being powers of $2$.\footnote{The
pseudocode of their algorithm does not use phases, but with some fiddling with
constants, it can be transformed into one that does.} They use
probabilistic arguments to argue that with high probability each node gathers
$q_k$ colors in a phase of length $\Theta(2^k \cdot \log^2 n)$. Instead, we
treat the number of colors gathered in a phase as a hard constraint (we require
that each node gathers $q_k$ of them), and instead the phase lengths of the
nodes become random variables. As it turns out, this subtle difference allows us
to derandomize the
algorithm.

Another small difference concerns the color selection. While \RAND chooses the
color uniformly at random from the set $R_{k^*}$, the algorithm
of~\cite{EmGoKa19} would choose it from the set~$\uplus_{j=0}^{k^*} R_j$. This
difference affects only the constant factors of the analysis.

\subparagraph{Gain of RAND.}

As mentioned earlier, the gain of \RAND is directly ensured by the algorithm
definition \emph{provided} that each node has completed some number of phases. 

Note that $q_k$ is slightly less than $|R_k| = 2^k$; this gives a better bound
on the the expected phase lengths, while still ensuring that if all nodes
completed phase $\ell$, they gathered at least $2^{\ell-1}$ \emph{common}
colors.

\begin{lemma}
\label{lem:det_gain}
    If every node has completed phase $\ell$, then $\RAND(E) \geq 2^{\ell-1}$.
\end{lemma}

\begin{proof}
    In phase $\ell$, each node gathers at least $q_\ell$ colors from the palette
    $R_\ell$, i.e., all colors from~$R_\ell$ except at most $|R_\ell| - q_\ell
    \leq 2^\ell / (2n)$ colors. Thus, all nodes share at least $|R_\ell| - n \cdot
    (2^\ell / (2n)) = 2^{\ell-1}$ common colors from $R_\ell$, i.e., $\RAND$
    fully uses at least $2^{\ell-1}$ colors.
\end{proof}

Note that we could sum the above bound over all phases completed by all nodes,
but this would not change the asymptotic gain. 

The above lemma points us to the goal: to show that for an input $E$ with a
sufficiently large $\delta(\I)$, each node completes an appropriately large
number of phases. In \autoref{sec:random_event_T}, we show that if $\delta(\I) =
\Omega(2^\ell \cdot \log^2 n)$, then each node completes $\ell$ phases,
\emph{provided} a certain random event $\T$ occurs.


\subsection{Constructing the Potential: Insights and Definitions}
\label{sec:introducting_potential}

In the field of online algorithms, the derandomization has been successfully
conveyed several times by replacing the method of conditional probabilities by
an \emph{online-computable potential function} that guides the choice of the
deterministic algorithm~\cite{AlAABN09,BiFeSc21,BucNao09a,LeUmXi23}.

This method is based on the following framework. First, introduce $\ell$ random
expressions~$\set{Z_i}_{i=1}^\ell$ to be controlled. Second, define a potential
function $\Phi = \sum_{i=1}^\ell \exp(Z_i)$. Third, show that the random actions
of an algorithm at each step decrease $\exp(Z_i)$ (for each~$i$) in expectation,
which implies that $\Phi$ decreases as well. (By the probabilistic method, this
implies the existence of a \emph{deterministic action} of an algorithm in a step
that does not increase~$\Phi$.) Finally, by the non-negativity of the
exponential function, $\exp(Z_i)$ is bounded by the initial value $\Phi^0$ of
the potential (in each step), which shows that $Z_i$ can always be bounded by
$\ln(\Phi^0)$. This process can be seen as a derandomization of the Chernoff
bound / high probability argument.

In order to apply this very general framework, we must overcome several
technical difficulties, which we explain below. Except for the definition of
$\Phi$ (and related definitions of~$w_{ik}, c_{ik}$, $d_k$ and $Z_i$), the
discussion in this section is informal and will not be used in formal proofs
later.

\subparagraph{Variables to be Controlled.}

The first step is to identify the variables to control. Natural choices are the
node degrees and the number of colors gathered so far by each node. 

For this purpose, we define $c_{ik} = |C_{ik}|$ for each node $i$ and each phase
$k \geq 0$.

We also introduce counters $w_{ik}$, which are initially set to zero. Recall
that whenever a~new hyperedge~$S$ containing node $i$ arrives, \RAND computes
$\plow = \min_{i \in S} p(i)$. For each node~$i \in S$ such that $p(i) \leq
\plow + h - 1$, we increment the counter $w_{i,p(i)}$. Note that these
$w_{i,p(i)}$ variables are incremented exactly for nodes $i$ for which there is
a non-zero probability of increasing their set of colors~$C_{i,p(i)}$ (as then a
random integer $k^*$ chosen by \RAND has a~non-zero chance of being equal
to~$p(i)$).

By the definition of $w_{ik}$, at each step $\sum_{k \geq 0} w_{ik} \leq
\deg(i)$. While these quantities are not necessarily equal, we will treat
$\sum_{k \geq 0} w_{ik}$ as a good proxy for $\deg(i)$ and deal with the
discrepancy between these two quantities later.

\subparagraph{Linking the Variables.}

Now we want to introduce an expression that links $\sum_{k \geq 0} w_{ik}$ (the
proxy for degree) with $\sum_{k \geq 0} c_{ik}$ (the number of colors gathered)
for a node $i$. A simple difference of these two terms does not make sense: the
expected growth of $c_{ik}$ varies over time, since it is easier to gather new
colors when $c_{ik}$ is small. This effect has been studied in the context of
the coupon collector's problem~\cite{MitUpf17}. To overcome this issue, we
introduce the following function, defined for each integer $k \geq 0$:
\begin{align*}
    d_k(m) & \triangleq h \cdot \sum_{j=1}^{m} \frac{2^k}{2^k-j+1}  
        && \text{(defined for each $m \leq 2^k$)}
\end{align*}
Note that in a process of choosing random colors from palette $R_k$ of $2^k$
colors, the expected number of steps till $m$ different colors are gathered is
$d_k(m) / h$. 

Now we focus on a single node $i$ in phase $p(i)$. We consider a sequence of
hyperedges~$S$ containing node $i$, neglecting those hyperedges $S$ for which
$p(i) \geq \plow + h$. That is, all considered hyperedges increment the counter
$w_{i,p(i)}$. Then, the value of $d_k(c_{i,p(i)})$ corresponds to the expected
number of such hyperedges needed to gather $c_{i,p(i)}$ colors from the
palette~$R_{p(i)}$. We can thus use the expression $\sum_{k \geq 0} (w_{ik} - 2
\cdot d_k(c_{ik}))$ to measure the progress of node~$i$: small values of this
expression indicate that it is gathering colors quite fast, while large values
indicate that it is falling behind. Note that since $c_{ik} \leq q_k \leq 2^k$,
the value of $d_k(c_{ik})$ is always well defined. 

We note that the previous applications of the potential function
method~\cite{AlAABN09,BiFeSc21,BucNao09a,LeUmXi23} did not require such
transformations of variables: in their case, the potential function was used to
guide deterministic \emph{rounding}: the function~$\Phi$ directly compared the
cost (or gain) of a deterministic algorithm with that of an online fractional
solution. Instead, our solution operates directly on the input, without the need
to generate a fractional solution first.

\subparagraph{Scaling.}

Using the random choices of \RAND, we can argue that in expectation the value 
of~$Z^*_i \triangleq \sum_{k \geq 0} (w_{ik} - 2 \cdot d_k(c_{ik}))$ is decreasing
in time. However, to argue that $\exp(Z^*_i)$ is also decreasing in expectation,
we would have to ensure that $Z^*_i$ is upper-bounded by a~small constant (and
use the fact that $\exp(x)$ can be approximated by a linear function for small
$x$).

In the previous papers~\cite{AlAABN09,BiFeSc21,BucNao09a,LeUmXi23}, this
property was achieved by scaling down $Z^*$ by the value of \OPT. An algorithm
was then either given an upper bound on \OPT (in the case of the throughput
maximization of the virtual circuit routing~\cite{BucNao09a}) or \OPT was
estimated by standard doubling techniques (in the case of the cost minimization
for set cover variants~\cite{AlAABN09,BiFeSc21,LeUmXi23}). In the latter case,
the algorithm was restarted each time the estimate on \OPT was doubled.
Unfortunately, the DSC problem (which is an unbounded-gain maximization problem)
does not fall into any of the above categories, and the doubling approach does
not seem to work here. 

Instead, we replace the scaling with a weighted average. That is, for 
each node $i$, we define 
\[
    Z_{i} \triangleq 
    \sum_{k \geq 0} \frac{w_{ik} - 2 \cdot d_k(c_{ik})}{4 h \cdot 2^k}
\]
and the potential as
\[
    \Phi \triangleq \sum_{i \in [n]} \exp(Z_i).
\]

Note that all counters and variables defined above are random variables; they
depend on particular random choices of $\RAND$. We use $p^t(i)$, $\deg^t(i)$,
$w^t_{ik}$, $c^t_{ik}$, $Z^t_i$ and $\Phi^t$ to denote the values of the
corresponding variables at the end of step $t$ of the algorithm (after \RAND has
processed the hyperedge presented in step $t$). The value of $t = 0$ corresponds
to the state of these variables at the beginning of the algorithm; note that
$p^0(i) = \deg^0(i) = w^0_{ik} = c^0_{ik} = Z^0_i = 0$ for all $i$ and $k$.
Therefore,
\begin{equation}
\label{eq:potential_initial}
    \Phi^0 = n.
\end{equation}

\subparagraph{Random event $\T$.}

For an input instance $E$ consisting of $T$ steps, we define a random event $\T$
that occurs if $\Phi^t \leq n$ for each step $t \in \set{0, \dots, T}$ of \RAND
execution on input~$\I$.


\subsection{Relating Potential to Algorithm Performance}
\label{sec:random_event_T}

We begin by presenting the usefulness of the event $\T$. We emphasize that the
following lemma holds for all executions of \RAND in which the event $\T$
occurs. Its proof is deferred to \autoref{sec:phases}.

\begin{restatable}{lemma}{deltaphases}
\label{lem:delta_to_phases}
    Fix a sequence $\I$ such that $\delta(E) > 24 h \cdot \ln(4\e \cdot n) \cdot
    2^\ell$ for some integer $\ell \geq 0$. If $\T$ occurs, then each node has
    completed its phase $\ell$.
\end{restatable}

Using the lemma above, we can relate the gain of \RAND on \emph{an arbitrary}
input $\I$ to that of \OPT, if only $\T$ occurs.

\begin{lemma}
\label{lem:competitive_ratio}
    Fix a sequence $\I$. If $\T$ occurs, then $\RAND(\I) \geq \OPT(\I) / (96 h
    \cdot \ln(4\e \cdot n)) - 1/4$.
\end{lemma}

\begin{proof}
    Let $r \triangleq 24 h \cdot \ln(4\e \cdot n)$. We consider two cases.
    \begin{itemize}
    \item First, we assume that $\delta(E) > r$. Then we can find an integer $\ell
        \geq 0$ such that $r \cdot 2^\ell < \delta(E) \leq r \cdot 2^{\ell+1}$.
        By \autoref{lem:delta_to_phases}, each node then completes its phase
        $\ell$, and so by \autoref{lem:det_gain}, $\RAND(\I) \geq 2^{\ell-1}
        \geq \delta(E) / (4 r)$.
    \item Second, we assume that $\delta(E) \leq r$. Trivially, $\RAND(\I) \geq 0
        \geq (\delta(E) - r) / (4 r)$.
    \end{itemize}
    In both cases, $\RAND(\I) \geq (\delta(E) - r) / (4 r) \geq \OPT(E) / (4r) -
    1/4$. 
\end{proof}


\subsection{Derandomization of RAND}

To analyze the evolution of $\Phi$, we note that $\Phi^t$ (and also other
variables $w^t_{ik}$, $c^t_{ik}$ or $Z^t_i$) depends only on the random choices of
\RAND till step $t$ (inclusively).  Moreover, $\set{\Phi^t}_{t \geq 0}$ is 
a~supermartingale with respect to the random choices of \RAND in consecutive steps.
Specifically, the following lemma holds; its proof is deferred to
\autoref{sec:potential_monotonicity}.

\begin{restatable}{lemma}{supermartingale}
\label{lem:supermartingale}
    Fix a step $t$ and an outcome of random choices till step $t-1$ inclusively.
    (In particular, this will fix the value of $\Phi^{t-1}$.) Then, $\E[\Phi^t]
    \leq \Phi^{t-1}$, where the expectation is taken exclusively over random
    choices of \RAND in step $t$.
\end{restatable}

The above lemma states that the value of $\Phi$ is non-increasing in
expectation. In fact, an~inductive application of this lemma shows that
$\E[\Phi^t] \leq \Phi^0 = n$. However, this does not imply that $\T$ occurs with
a reasonable probability, especially when the input length is large.\footnote{By
Markov's inequality, for a chosen step $t$, $\Phi^t \leq 2 n$ holds with
probability at least $1/2$. While such a~relaxed bound on $\Phi^t$ would be
sufficient for our needs, in our proof, we need such a bound to hold for all steps
$t$ simultaneously.}

However, using \autoref{lem:supermartingale}, we can easily derandomize $\RAND$
using the method of conditional probabilities using potential $\Phi$ as an
online-computable counterpart of a pessimistic estimator. To do this, we proceed
iteratively, ensuring at each step $t$ that $\Phi^t \leq n$. This is trivial at
the beginning, since $\Phi^0 = n$ by \eqref{eq:potential_initial}. 

Suppose we have already fixed the random choices of \RAND till step $t-1$
inclusively and have $\Phi^{t-1} \leq n$. Consider a hyperedge $S$ presented in
step $t$. Note that all other variables indexed by $t-1$, such as $p^{t-1}(i)$,
are also fixed. Then \autoref{lem:supermartingale} states that $\E[\Phi^t] \leq
\Phi^{t-1} \leq n$. That is, the random choice of a color at step $t$ guarantees
that $\E[\Phi^t] \leq n$. This choice is made from a finite and well-defined set
of colors $\R \triangleq \biguplus_{\plow \leq k \leq \plow+h-1} R_k$, where
$\plow = \min_{i \in S} p^{t-1}(i)$.

By the probabilistic method, at each step $t$, there exists a \emph{deterministic}
choice of a~color from~$\R$ that ensures that $\Phi^t \leq n$. The resulting
algorithm is called \DET. (If more than one color leads to $\Phi^t \leq n$,
$\DET$ chooses any of them.) Since $|\R|$ is bounded by a~polynomial of $n$ and
$|E|$, \DET can be implemented in polynomial time by simply checking all
possible colors of $\R$.

\begin{theorem}
\label{thm:main}
    \DET is $O(\log^2 n)$-competitive for the DSC problem.
\end{theorem}

\begin{proof}
    As described above, $\DET$ guarantees that $\Phi^t \leq n$ for each step $t
    \in \set{0, \dots, T}$, i.e.,~$\T$~occurs. Note that
    \autoref{lem:competitive_ratio} lower-bounds the gain of \RAND in every
    execution conditioned on $\T$, and \DET can be seen as such an
    execution. Therefore, the bound of \autoref{lem:competitive_ratio} can be
    applied, yielding
    \[
        \DET(\I)
            \geq \frac{\OPT(\I)}{96 \cdot h \cdot \ln(4\e \cdot n)} - \frac{1}{4},
    \]
    i.e., the competitive ratio of \DET is at most $96 \cdot h \cdot \ln(4\e
    \cdot n) = O(\log^2 n)$.

    Note that, by the lower bounds of~\cite{PaBaVa15,EmGoKa19}, an additive term
    (in our case equal to $1/4$) is inevitable for obtaining a sub-linear
    competitive ratio.
\end{proof}


\section{Bounding Number of Phases}
\label{sec:phases}

In this section, we fix an input sequence $E$ consisting of $T$ steps. Our goal
is to estimate the number of phases of nodes in the execution of \RAND,
conditioned on the random event~$\T$, i.e., to prove
\autoref{lem:delta_to_phases}. To this end, we consider a node $i$, assume that
it has completed $\ell$~phases, and show an upper bound on the final degree of
$i$ as a function of $\ell$.

\subparagraph{Bounding Variables w Using Potential.}

Recall that in some steps where the degree of a node $i$ grows, 
the counter $w_{i,p(i)}$ is incremented. Thus, our first goal is to 
upper-bound values of these counters at the end of the execution of \RAND.

Below, $H(m)$ denotes the $m$-th harmonic number.
The following technical claim is proved in \autoref{sec:technical}.

\begin{restatable}{claim}{harmonic}
\label{cla:harmonic}
    For each $k \geq 0$ it holds that $H(2^k) - H(2^k-q_k) \leq \ln (4 \e \cdot n)$.
\end{restatable}

\begin{lemma}
\label{lem:d_q_k_bound}
    Fix a step $t \leq T$, a node $i$ and a phase $k \geq 0$. Then,
    $d_k(c^t_{ik}) \leq h \cdot \ln(4\e \cdot n) \cdot 2^k$.
\end{lemma}

\begin{proof}
    Note that at all times, $c^t_{ik} \leq q_k$. 
    Using the definition of $d_k$,
    \begin{align*}
        d_k(c_{ik}) \leq d_k(q_k) 
        & = h \cdot \sum_{j=1}^{q_k} \frac{2^k}{2^k-j+1} \\ 
        & = h \cdot \left( H(2^k) - H(2^k-q_k) \right) \cdot 2^k \\
        & \leq h \cdot \ln(4 \e \cdot n) \cdot 2^k.
            && \text{(by \autoref{cla:harmonic})}
    \qedhere
    \end{align*}
\end{proof}

\begin{lemma}
\label{lem:steps_bound}
    Fix a node $i$ and a phase $\ell \geq 0$. If the event $\T$ occurs, then
    $\sum_{k=0}^\ell w^T_{ik} \leq 8 h \cdot \ln(4\e \cdot n) \cdot 2^\ell$.
\end{lemma}

\begin{proof}
    Fix the last step $t \leq T$ in which $\sum_{k=0}^\ell w_{ik}$ increases. By
    the choice of $t$, we have $w^t_{ik} = c^t_{ik} = 0$ for every phase $k > \ell$.

    Since $\T$ occurs, $n \geq \Phi^t = \sum_{j \in [n]} \exp(Z^t_j)$. Due to
    the non-negativity and monotonicity of the exponential function, $Z^t_i \leq
    \ln n$. Using the definition of $Z^t_i$, we get
    \begin{align*}
        \ln n \geq Z^t_i 
        & = \sum_{k \geq 0} \frac{w^t_{ik} - 2 \cdot d_k(c^t_{ik})}{4 h \cdot 2^k} \\
        & = \sum_{k=0}^\ell \frac{w^t_{ik} - 2 \cdot d_k(c^t_{ik})}{4 h \cdot 2^k} \\
        & \geq \frac{1}{4h \cdot 2^\ell} \cdot 
            \left( \sum_{k=0}^\ell w^t_{ik} 
                - 2 \cdot \sum_{k=0}^\ell h \cdot \ln(4\e \cdot n) \cdot 2^k 
            \right).
            && \text{(by \autoref{lem:d_q_k_bound})}
    \end{align*}
    Hence, again by the choice of $t$, 
    \begin{align*}
    \sum_{k=0}^\ell w^T_{ik} = \sum_{k=0}^\ell w^t_{ik} 
    & \leq 4 h \cdot 2^\ell \cdot \ln n + 4 h \cdot \ln(4\e \cdot n) \cdot 2^\ell
        < 8 h \cdot \ln(4\e \cdot n) \cdot 2^\ell.
    \qedhere
    \end{align*}
\end{proof}

\subparagraph{Bounding Node Degrees.} 

Recall that whenever a new hyperedge $S$ containing node $i$ arrives, $\plow =
\min_{i \in S} p(i)$ is determined. Then, for each node $i \in S$, if $p(i) \leq
\plow + h - 1$, the counter $w_{i,p(i)}$ is incremented. If $p(i) \geq \plow + h$,
the degree of $i$ grows, but $w_{i,p(i)}$ is not incremented. To estimate the
degree of $i$, we therefore introduce the counters $s_{ik}$, which are incremented in
the latter case. That is, for each node $i \in S$, we always have $\sum_{k \geq 0} (w_{ik} +
s_{ik}) = \deg(i)$.

The growth of the variables $s_{ik}$ is not controlled by the potential, but
they grow only for nodes whose degree is very high compared to the current
minimum degree. By constructing an~appropriate charging argument, we can link
their growth to the growth of other variables~$w_{ik}$.

\begin{lemma}
    \label{lem:w_star_bound_helper}
    Fix a step $t \leq T$, a node $i$, and a phase $\ell \geq 0$. Then, 
    $s^t_{i\ell} \leq \sum_{j \in [n]} \sum_{r=0}^{\ell-h} w^t_{jr}$.
\end{lemma}

\begin{proof}
    We fix node $i$, phase $\ell \geq 0$, and show the lemma by induction on
    $t$. The inductive basis is trivial, as $s^0_{i\ell} = 0 = \sum_{j \in [n]}
    \sum_{r=0}^{\ell-h} w^0_{jr}$.
    
    Now suppose that the lemma statement holds for step $t-1$. We look at how
    both sides of the inequality change as we increase the step superscripts
    from $t-1$ to $t$, and argue that the increase of the right hand side is at
    least as large as the increase of the left hand side. If $s_{i\ell}$ does
    not change, the inductive claim follows trivially. Otherwise, $s_{i\ell}$ is
    incremented by $1$. This happens only if $\ell = p(i)$, $i \in S$, and $\ell
    \geq \plow + h$. By the definition of $\plow$, this means that there exists
    at least one node $j^* \in S$ such that $p(j^*) = \plow$, and the
    corresponding counter~$w_{j^*,\plow}$ is also incremented. This means that
    the right hand side of the lemma inequality is incremented by at least
    $\sum_{j \in [n]} \sum_{r=0}^{\ell-h} (w^t_{jr} - w^{t-1}_{jr}) \geq
    w^{t}_{j^*,\plow} - w^{t-1}_{j^*,\plow} = 1$, and the inductive claim
    follows.
\end{proof}

\begin{lemma}
    \label{lem:w_star_bound}
    Fix a node $i$ and a phase $\ell \geq 0$. 
    If the event $\T$ occurs, then 
    $\sum_{k=0}^\ell s^T_{ik} \leq 16 h \cdot \ln(4\e \cdot n) \cdot 2^\ell$. 
\end{lemma}

\begin{proof}
    By \autoref{lem:w_star_bound_helper}, 
    \begin{align*}
        s^T_{ik} 
        \leq \sum_{j \in [n]} \sum_{r=0}^{k-h} w^T_{jr} 
        & \leq n \cdot 8 h \cdot \ln(4\e \cdot n) \cdot 2^{k-h} 
            && \text{(by \autoref{lem:steps_bound})} \\ 
        & \leq 8 h \cdot \ln(4\e \cdot n) \cdot 2^k.
            && \text{(as $h = \ceil{ \log n }$)} 
        \end{align*}
    Hence, $\sum_{k=0}^\ell s^T_{ik}  < 16 h \cdot \ln(4\e \cdot n) \cdot
    2^\ell$. 
\end{proof}

Finally, we can show \autoref{lem:delta_to_phases}, restated below.

\deltaphases*

\begin{proof}
    Suppose for a contradiction that there exists a node $i$ for which $p(i)
    \leq \ell$ at the end of the input. Then,
    \[
        \delta(E) 
        \leq \deg^T(i) = \sum_{k \geq 0} \brk*{ w^T_{ik} + s^T_{ik} } 
        = \sum_{k=0}^\ell \brk*{ w^T_{ik} + s^T_{ik} } 
        \leq 24 h \cdot \ln(4\e \cdot n) \cdot 2^\ell,
    \]
    where the last inequality follows by \autoref{lem:steps_bound} and
    \autoref{lem:w_star_bound}. This would contradict the assumption of the
    lemma.
\end{proof}


\section{Controlling the Potential}
\label{sec:potential_monotonicity}

In this section, we show that $\set{\Phi^t}_{t \geq 0}$ is a supermartingale
with respect to the choices of \RAND made in corresponding steps, i.e., we prove
\autoref{lem:supermartingale}.

Throughout this section, we focus on a single step $t$, in which \RAND processes
a hyperedge~$S$. Recall that \RAND first chooses a random integer $k^*$
uniformly from the set $\set{ \plow, \plow+1, \dots, \plow + h - 1 }$. Second,
conditioned on the first choice, it chooses a~random color uniformly from the
set $R_{k^*}$.

By the definition of our variables, for each node $i$ and each integer $k
\geq 0$, it holds that $w^t_{ik} - w^{t-1}_{ik} \in \set{0,1}$ and $c^t_{ik} -
c^{t-1}_{ik} \in \set{0,1}$.

\begin{lemma}
    \label{lem:c_ik_growth}
    Fix a node $i \in S$ and let $p = p^{t-1}(i)$. If $p \leq \plow + h - 1$,
    then $c^t_{ip} = c^{t-1}_{ip} + 1$ 
    with probability $(2^p-c^{t-1}_{ip})/(h \cdot 2^p)$.
\end{lemma}

\begin{proof}
    By the definition of $\plow$, we have $p \geq \plow$. Combining this with
    the lemma assumption, we get $p \in \set{ \plow, \dots, \plow + h - 1}$. 

    Note that $c^t_{ip} = c^{t-1}_{ip} + 1$ when node $i$ gathers a new color
    from $R_p$, and $c^t_{ip} = c^{t-1}_{ip}$ otherwise. For a node $i$ to
    gather a new color from $R_p$, first the integer $k^*$ chosen randomly from
    the set $\set{\plow, \dots, \plow+h-1}$ must be equal to~$p$, which happens
    with probability~$1/h$. Second, conditioned on the former event, a~color
    chosen randomly from the palette~$R_p$ must be different from all
    $c^{t-1}_{ip}$ colors from $R_p$ gathered so far by node $i$, which happens
    with probability $(|R_p|-c^{t-1}_{ip})/|R_p| =
    (2^p-c^{t-1}_{ip})/2^p$. Hence, the probability of gathering a~new color is
    $(2^p-c^{t-1}_{ip})/(h \cdot 2^p)$.
\end{proof}

We emphasize that the relations involving random variables in the following lemma (e.g.,~the 
statements such as~$Z^t_i \leq Z^{t-1}_i$) hold for all random choices made by \RAND.

\begin{lemma}
\label{lem:Delta_Z_bounded}
Fix a node $i$. Let $p = p^{t-1}(i)$. Then, either $Z^t_i \leq Z^{t-1}_i$ or
\[
    Z^t_i \leq Z^{t-1}_i + \frac{1}{4h \cdot 2^p} 
    + \begin{cases}
        -1/(2h \cdot 2^p \cdot \alpha) & \text{with probability $\alpha$,} \\
        0 & \text{otherwise.}
    \end{cases} 
\]
where $\alpha = (2^p - c^{t-1}_{ip})/(h \cdot 2^p)$. The probability is
computed exclusively with respect to random choices of \RAND in step $t$.
\end{lemma}

\begin{proof}
    For brevity, for an integer $k \geq 0$, we define 
    \begin{align*}
        \Delta w_{ik} &= w^t_{ik} - w^{t-1}_{ik}, \\
        \Delta d_k(c_{ik}) &= d_k(c^t_{ik}) - d_k(c^{t-1}_{ik}).
    \end{align*}
    As $c^t_{ik} \geq c^{t-1}_{ik}$ and $d_k$ is a non-decreasing function, we
    have $\Delta d_k(c_{ik}) \geq 0$.

    Let $S$ be the hyperedge presented in step $t$. By the definition of
    the variables~$w_{ik}$ (cf.~\autoref{sec:introducting_potential}), we have 
    \[
        \Delta w_{ik} = \begin{cases}
            1 & \text{if $i \in S$ and $k = p$ and $p \leq \plow + h - 1$,} \\
            0 & \text{otherwise.}
        \end{cases}    
    \]
    Now we consider two cases. 

    \begin{itemize}
    \item It holds that $i \notin S$ or $p \geq \plow + h$. Then, 
    \[
        Z^t_i - Z^{t-1}_i = \sum_{k \geq 0} \frac{ \Delta w_{ik} - 2 \cdot \Delta d_k(c_{ik})}{ 4h \cdot 2^k }
        \leq \sum_{k \geq 0} \frac{ \Delta w_{ik}}{ 4h \cdot 2^k } = 0.
    \]
    \item It holds that $i \in S$ and $p \geq \plow + h - 1$. Then, 
    $\Delta w_{ip} = 1$ and $\Delta w_{ik} = 0$ for $k \neq p$. Therefore,
    \begin{align*}
        Z^t_i - Z^{t-1}_i
        & = \sum_{k \geq 0} \frac{ \Delta w_{ik} - 2 \cdot \Delta d_k(c_{ik}) }{ 4h \cdot 2^k } \\
        & = \frac{ \Delta w_{ip} - 2 \cdot \Delta d_p(c_{ip}) }{ 4h \cdot 2^p } 
         + \sum_{k \neq p} \frac{ \Delta w_{ik} - 2 \cdot \Delta d_k(c_{ik}) }{ 4h \cdot 2^k } \\
        & \leq \frac{1}{ 4h \cdot 2^p } 
            - \frac{ \Delta d_p(c_{ip})}{ 2h \cdot 2^p }.
    \end{align*}
    To complete the proof, it now suffices to argue that $\Delta d_p(c^t_{ip}) =
    1/\alpha$ with probability $\alpha$ and $0$ with the remaining probability.
    This follows immediately by \autoref{lem:c_ik_growth}: With probability
    $\alpha$ we have $c^t_{ip} = c^{t-1}_{ip} + 1$, and hence $\Delta d_p(c_{ip})
    = d_p(c^t_{ip}) - d_p(c^{t-1}_{ip}) = h \cdot 2^p / (2^p - c^t_{ip} + 1) = h
    \cdot 2^p / (2^p - c^{t-1}_{ip}) = 1/\alpha$. With the remaining probability
    $c^t_{ip} = c^{t-1}_{ip}$ and thus $\Delta d_p(c_{ip}) = 0$.
    \qedhere
    \end{itemize}
\end{proof}

For the final lemma, we need the following technical bound (proven in
\autoref{sec:technical}). This can be seen as a reverse Jensen's type inequality. 

\begin{restatable}{claim}{expbound}
\label{cla:exp_bound}
    Fix $\eps \in [0,1]$, $\alpha \in [\eps, 1]$ and a real $x$.
    Let $X$ be a random variable such that
    \[
        X = \begin{cases}
            x - \eps / \alpha & \text{with probability $\alpha$,} \\
            x & \text{otherwise.} 
        \end{cases}
    \]
    Then, $\E[\e^X] \leq \e^{x-\eps/2}$.
\end{restatable}

Finally, we can prove \autoref{lem:supermartingale}, restated below.

\supermartingale*

\begin{proof}
    Fix a node $i$. We will show that 
    \begin{equation}
        \label{eq:exp_Z_bound}
        \E[\exp(Z^t_i)] \leq \exp(Z^{t-1}_i).
    \end{equation}
    The lemma then follows by summing the above inequality over all nodes.

    If $Z^t_i \leq Z^{t-1}_i$, \eqref{eq:exp_Z_bound} follows trivially.
    Otherwise, \autoref{lem:Delta_Z_bounded} implies that
    \[ 
        Z^t_i \leq Z^{t-1}_i + \frac{1}{4h \cdot 2^p} + \begin{cases}
            -1/(2h \cdot 2^p \cdot \alpha) & \text{with probability $\alpha$,} \\
            0 & \text{otherwise.}
        \end{cases} 
    \]
    where $p = p^{t-1}(i)$ and $\alpha = (2^p - c^{t-1}_{ip})/(h \cdot 2^p)$. As
    the random choices are fixed until step~$t-1$, the variables $Z^{t-1}_i$, $p$ and
    $\alpha$ are no longer random variables, but real numbers.
    
    Note that $c^{t-1}_{ip} \leq q_p - 1 \leq 2^p -1$ as otherwise phase $p$ of
    node $i$ would have ended in an~earlier step. Hence, $\alpha = (2^p -
    c^{t-1}_{ip})/(h \cdot 2^p) \geq 1/(2h \cdot 2^p)$.

    Thus, $x = Z^{t-1}_i + 1/(4h \cdot 2^p)$, $\eps = 1 / (2h \cdot 2^p)$, and
    $\alpha$ satisfy the conditions of \autoref{cla:exp_bound}. (In particular,
    $\eps \leq \alpha \leq 1$.) This claim now yields
    \[
        \label{eq:exp_bound}
        \E[\exp(Z^t_i)] \leq \exp\left( Z^{t-1}_i + \frac{1}{4h \cdot 2^p} 
            - \frac{1}{2} \cdot \frac{1}{2h \cdot 2^p} \right)
            = \exp(Z^{t-1}_i),      
    \]
    which concludes the proof of \eqref{eq:exp_Z_bound}, and thus also the lemma.
\end{proof}


\section{Conclusions}

In this paper, we have constructed a deterministic $O(\log^2 n)$-competitive
algorithm for the disjoint set covers (DSC) problem. Closing the remaining
logarithmic gap between the current upper and lower bounds is an interesting
open problem that seems to require a new algorithm that goes beyond the
phase-based approach. 

We have developed new derandomization techniques that extend the existing
potential function methods. We hope that these extensions will be useful for
derandomizing other online randomized algorithms, and eventually for providing a
coherent online derandomization toolbox.

\section*{Acknowledgements}

The authors are grateful to the anonymous reviewers for their insightful comments.


\bibliography{references}


\appendix

\section{Proofs of Technical Claims}
\label{sec:technical}

\harmonic*

\begin{proof}
    Assume first that $2^k < 4n$. As 
    $H(m) \leq 1 + \ln m = \ln (\e \cdot m)$ for each $m$, we get
    $H(2^k) - H(2^k-q_k) < H(4n) \leq \ln(4\e \cdot n)$ and the lemma follows.

    Hence, in the following, we assume that $2^k \geq 4n$, obtaining
    \begin{align*}
        2^k-q_k
        & = 2^k - \lceil \brk{1-1/(2n)} \cdot 2^k \rceil 
            && \text{(by the definition of $q_k$)} \\
        & \geq 2^k - \left(1-1/(2n) \right) \cdot 2^k - 1 \\
        & = 2^k/(2n) - 1 \\
        & \geq 2^k / (4n).
            && \text{(as $2^k \geq 4n$)} 
    \end{align*}  
    As $q_k \leq 2^k$,
    we can use the relation $H(m) \geq \ln m$ holding for every $m \geq 0$, 
    obtaining 
    \[
        \label{eq:harmonic_bound}
        H(2^k) - H(2^k-q_k) \leq \ln(\e \cdot 2^k) - \ln(2^k - q_k)
        \leq \ln(\e \cdot 2^k) - \ln(2^k / (4n)) = \ln(4 \e \cdot n).
        \qedhere
    \]
\end{proof}

\expbound*

\begin{proof}
Using $\alpha \geq \eps$, we obtain
\begin{equation}
\label{eq:crucial_1}
    (1+\eps/\alpha) \cdot (1 - \eps/(2\alpha)) = 1 + \eps/(2\alpha) - \eps^2/(2\alpha^2) \geq 1.
\end{equation}
Next, we use the relation $\e^{-y} \leq (1+y)^{-1}$ (holding for every $y > -1$) 
and \eqref{eq:crucial_1} to argue that
\begin{equation}
\label{eq:crucial_2}
    \exp(- \eps / \alpha) 
    \leq (1+\eps/\alpha)^{-1}
    \leq 1-\eps/(2\alpha).
\end{equation}
Finally, this gives
\begin{align*}
    \E[\e^X] 
    & = \alpha \cdot \exp(x-\eps/\alpha) + (1-\alpha) \cdot \exp(x) \\
    & = \e^{x} \cdot \left(\alpha \cdot \exp(-\eps/\alpha) + 1-\alpha  \right) \\
    & \leq \e^{x} \cdot \left(\alpha - \eps/2 + 1-\alpha  \right) && \text{(by \eqref{eq:crucial_2})} \\
    & \leq \e^{x-\eps/2}. && \text{(as $1-\eps/2 \leq \e^{-\eps/2}$ for each $\eps$)}
\qedhere
\end{align*}
\end{proof}

\end{document}